\documentclass[12pt,a4paper,thmsb]{article}

\usepackage{eurosym}
\usepackage{amsmath}
\usepackage{amssymb}
\usepackage{amsfonts}
\usepackage{color}
\usepackage{float,graphicx,xcolor,tikz,caption,subcaption}

\newtheorem{lemma}{Lemma}
\newtheorem{proposition}{Proposition}
\newtheorem{corollary}{Corollary}

\newenvironment{proof}[1][Proof]{\noindent\textbf{#1.} }{\ \rule{0.5em}{0.5em}}

\begin{document}

\title{Domestic Competitive Balance and International Success: The Case of The Football Industry\thanks{{\footnotesize{The first author acknowledges financial support from the Spanish Agencia Estatal de Investigaci\'{o}n (AEI) through grant PID2020-115011GB-I00, funded by MCIN/AEI/10.13039/501100011033. }}
}}
\author{\textbf{Juan D. Moreno-Ternero}\thanks{Universidad Pablo de Olavide, Seville, Spain. Corresponding author. e-mail: jdmoreno@upo.es} \\
\textbf{Tim Pawlowski}\thanks{University of T\"{u}bingen, Germany. Tim Pawlowski is also affiliated with the LEAD Graduate School and Research Network, as well as the Interfaculty Research Institute for Sports and Physical Activity in T\"{u}bingen.} \\
\textbf{Shlomo Weber}\thanks{Southern Methodist University, Dallas, USA.}}

\maketitle

\begin{abstract}
\noindent  This paper examines the interdependence of international success and competitive balance of domestic sports competitions.  More specifically, we apply the notion of the Herfindahl-Hirschman index to examine the effect of international rewards on distortion of competitive balance in domestic competitions and derive conditions under which the level of domestic competitive balance raises or falls. Our results yield interesting policy implications for the regulation of prize schemes in international competitions.

\end{abstract}

%C72, D21, L83


\newpage

\section{Introduction}

The rapid expansion of the sports industry has been tremendously important in the last decades. Deloitte's 2023 sports industry outlook states that  \textit{the growth and transformation of the sports industry is forcing organizations to take a more sophisticated approach  that makes the  industry more attractive for investors, more immersive for fans, and more supportive of athletes}. 
In the scientific literature, there has been an increasing interest to study organizational aspects related to this sports industry. These range from the fair design of competitions and prize awards, to the analysis of game scheduling, or the allocation of gate and sportscast revenues (e.g., Peeters, 2015; Arlegi and Dimitrov, 2020; Berganti\~ nos and Moreno-Ternero, 2020; 2021; Csat\'{o}, 2021; Dietzenbacher and Kondratev, 2023; Li et al., 2023). %, the analysis of game scheduling (e.g., Dur\'{a}n et al., 2017, ) and transfer systems (e.g., Hoey et al., 2021), to the impact of managers (e.g., van Ours and van Tuijl, 2016; Peeters et al., 2020)% The reader is referred to Wright (2009, 2014) or Kendall and Lenten (2017) for surveys of the fast-expanding literature, and Palacios-Huerta (2014) or  for recent books.
And a special emphasis has been made on competitive balance, a core concept in the sports economics literature that usually refers to the degree of parity or equality among competing teams in the same league. However, the concept %of competitive balance 
is not limited to sports and has also been scrutinized in other parts of the entertainment business and beyond. For instance, the effects of competition on process innovation and product introduction have long been studied (e.g., Vives, 2008). In this regard, competitive balance ensures the absence of unfair advantages for a small set of firms. Likewise, in politics, competitive balance prevents a single political party to attain a monopoly on electoral power (e.g., Roemer, 2001). Also, equality of opportunity in education is linked to competitive balance (also known as \textit{leveling the playing field}) for the entire student population (e.g., Roemer, 2006). 

The seminal contributions in the context of sports by Rottenberg (1956) and Neale (1964) emphasize the positive relation between the level of balance within a sports competition and its attractiveness for spectators. Ever since there have been passionate discussions amongst scholars and practitioners alike as to whether (and, if so, to which extent) competitive balance indeed matters for spectators and fans.  
The empirical literature exploring this is vast and inconclusive. For instance, Butler (1995) revealed that the relationship between the number of spectators and competitive balance for Major League Baseball in the US crucially depends on the chosen measure of competitiveness. More recently, Guironnet (2023) found that French Ligue 1 attendance is positively correlated with competitive balance, whereas English Premier League attendance is more dependent of the presence of star players. Other recent studies, either following a behavioral economics paradigm (see Budzinski and Pawlowski, 2017, for an overview) and/or using more granular data for exploring in-play dynamics of entertainment utility (see seminal work by Ely et al., 2015, and Richardson et al., 2023 for an application to football) allow to conclude that league organizers, as well as competitors in sports markets, should have a common incentive to preserve a certain level of competitive balance.

Our research question is motivated by the regulatory practice in professional sports. For instance, North American professional sports have long been implementing salary caps or draft systems, in which weaker clubs have priority (at least a higher probability) to pick draftees each season to maintain or even increase competitive balance. Likewise, the allocation of revenues raised from broadcasting rights, the most important source of revenue for professional sports, is certainly driven (at least, partially) by that goal too (e.g., Berganti\~ nos and Moreno-Ternero, 2023a). But sometimes competitions overlap and the revenue allocations might intertwine each other. In particular, when participating clubs do not fully overlap in those competitions (such as in the case of the UEFA Champions League).\footnote{UEFA stands for the Union of European Football Associations, the governing body of football in Europe.} %This is precisely the case with international competitions, such as the UEFA Champions League, that involve some competitors from different countries (selected among those competing in each domestic competition, via some merit scheme). \footnote{In October 2018, it was announced that the Netherlands' three representatives in European football would distribute some of their earnings from these competitions among the other clubs in Eredivisie.} 
In this paper, we are concerned with the effect that revenue allocation schemes at international competition will have on domestic competition.

%As mentioned above, competitive balance is an important concept for professional sports. 
While one could safely argue that competitive balance is one of the key issues that European football has to address in order to ensure its long-term prosperity,\footnote{In fact, UEFA President Aleksander Ceferin told the general assembly of the European Club Association (ECA) in 2017: \textit{"Let's put our cards on the table and be honest with ourselves: the biggest challenge over the next few years will be `competitive balance'.}" https://www.reuters.com/article/us-soccer-europe-clubs/uefa-president-says-rich-poor-divide-is-key-issue-in-european-football-idUSKCN1BF2CX} its quantification is not an easy task. The difficulties are mainly associated with its multi-dimensional nature. Consequently, there are multiple approaches to measure competitive balance. As Zimbalist (2002, p. 112) puts it: ``there are almost as many ways to measure competitive balance as there are to quantify money supply". %In fact, Zimbalist (2002) offered his own measure of competitive balance which is based on the standard deviation of winning percentages from year to year, and it assumes that all teams are of equal strength.\footnote{That is, in a league of perfect parity, the standard deviation of win percentages would be $\sigma= \frac{0.5}{M}$, where $M$ is the number of games played during the season by each team.} This assumption is not entirely reasonable in our setting, where 
As we consider both leagues and teams of various strengths and put an emphasis on the budget distributions, there are two indices that fit our approach: The Concentration Ratio (CR), the fraction of industry size held by the larger firms (Hall and Tideman, 1967; Besanko et al., 2017), and the Herfindahl-Hirschman Index (HHI), the sum of the quadratic shares (Herfindahl, 1950; Hirschman, 1945).\footnote{Bossert et al. (2022) have recently studied a larger family of indices encompassing these two. See also Martinez and Moreno-Ternero (2022) for alternative families of indices in a different setting.} We opt for the latter because it is superior on analytical grounds (including the fact that it weights firms according to their size) and there is a long tradition of use in industrial organization, since the U.S. Horizontal Merger Guidelines were unveiled in 1982. Therein, the U.S. Department of Justice divides the spectrum of market concentration as measured by the HHI into three regions that can be broadly characterized as unconcentrated (HHI below 1000), moderately concentrated (HHI between 1000 and 1800) and highly concentrated (HHI above 1800). 

The importance of HHI, and, in particular, its link with the societal welfare has been studied extensively in the industrial organization literature (e.g., %Cowling and Waterson, 1976; 
Dansby and Willig, 1979; Farrell and Shapiro, 1990). More recently, it also has become a popular measure in sports economics (e.g., Owen et al., 2007).\footnote{Depken (1999) originally applied the same metric to estimate inequality in the baseball league.}   
In our paper, we apply the HHI to examine the effect of international rewards on distortion of competitive balance in domestic competitions and derive conditions under which the level of domestic competitive balance raises or falls. We focus on the case of football, the most popular sport (for which sources estimate over 3.5 billion fans worldwide).\footnote{https://worldpopulationreview.com/country-rankings/most-popular-sport-by-country} Our results derive interesting policy implications for the regulation of prize schemes in international competitions represented in our setting by the UEFA Champions League (whose structure is described in the following section).  

The rest of the paper is organized as follows. In Section 2, we describe the structure and brief history of the UEFA Champions League that serves as the platform for our analysis. In Section 3, we present the benchmark model, which examines the competitive balance implications of award schemes that share the prize endowment of the international competition $E$ equally among the $k$-top clubs in the domestic competition. We show that there is a threshold value of the number of award recipients, $k^*(E)$, such that awarding to $k<k^*(E)$ clubs reduces the level of competitive balance (relative to the pre-award stage) while awarding it to $k>k^*(E)$ clubs increases it. In particular, the award to only one top club hurts competitive balance the most, while sharing it equally among all clubs (in the domestic competition) improves the most. In Section 4, we focus our analysis on the threshold value $k^*(E)$ and show that this value is weakly increasing in the size of the prize $E$. However, for a given $k$, competitive balance follows an inverse U-shape with respect to $E$, i.e., it is increasing below some threshold value and declining above that threshold. We also provide the support of the equal-share scheme among the $k$-top schemes by showing that it is optimal in terms of competitive balance among all other $k$-top schemes. In Sections 5 and 6, we offer an illustration and discussion of our results, respectively. 

\section{The UEFA Champions League}

 %which will be our concern in this paper. %for which international and domestic competitions coexist. In this paper, we are precisely concerned with the interplay of those two types of competitions.

In 1955, the so-called European Champion Clubs' Cup was launched. It was a yearly competition played by the national league football  champions of the strongest UEFA national associations.\footnote{More precisely, sixteen teams played the first edition: Milan (Italy), AGF Aarhus (Denmark), Anderlecht (Belgium), Djurgrden (Sweden), Gwardia Warszawa (Poland), Hibernian (Scotland), Partizan (Yugoslavia), PSV Eindhoven (Netherlands), Rapid Wien (Austria), Real Madrid (Spain), Rot-Weiss Essen (West Germany), Saarbrucken (Saar), Servette (Switzerland), Sporting CP (Portugal), Stade de Reims (France), and Vrs Lobog (Hungary). %Some of these teams had not been champions in their respective league , such as clubs from Poland, Hungary, and Switzerland. 
The English champion, Chelsea, decided not to play, and there was no English team in that first
edition. Thus, it was intended for football champions, but the intention was also for all
UEFA national associations to participate, not only the ``strongest."} 
The competition took on its current name, the UEFA Champions League, in 1992, adding a group stage to the competition and allowing multiple entrants from certain countries. In its present format, it begins with four knockout qualifying rounds and a play-off round. The surviving clubs enter the group stage, joining clubs qualified in advance. Overall, 32 clubs are drawn into eight groups of four clubs and play each other in a double round-robin system. The eight group winners and eight runners-up proceed to the knockout phase that culminates with the final match. The 32 teams participating in the group stage during the last editions %of the UEFA Champions League, from 2018-19 to 2023-24, 
have been the title holders (of both UEFA Champions League and Europa League), the 10 champions from the ten strongest domestic competitions, the 6 runners-up from the six strongest domestic competitions, the 4 third-placed and the 4 fourth-placed teams from the four strongest domestic competitions, and the winners of a six-round qualifying tournament involving winners and, possibly, runner-ups from the remaining weaker domestic competitions. Thus, some domestic competitions are not represented in the group stage, whereas others are with a number of clubs ranging from 1 to 4 (occasionally 5, for one or two at most).\footnote{From season 2024/2025, 36 (instead of 32) qualified clubs will be organized in a single league which follows an unbalanced schedule, i.e., each club plays 8 matches. The top 8 clubs in this league will directly qualify for the round of 16. The remaining 8 spots for the knockout phase will be decided after two-legged play-offs between clubs ranked 9 to 24. This change in the format will increase the number of games played from 125 to 225, which will probably make the competition more lucrative for UEFA and the participating teams.}
%It is, by now, one of the most prestigious tournaments in the world and the most prestigious club competition in European football. The clubs that take part in each edition share more than two billion euros in payments from UEFA, with a meritocratic scheme rewarding teams as they advance in the competition.
%Real Madrid is the most successful club in the competition's history, having won the tournament 14 times, including its first five editions. \footnote{} 

Since 1992, Spanish clubs have accumulated the highest number of victories (19 wins), followed by England (15 wins), Italy (12 wins) and Germany (8 wins). England has the largest number of winning clubs, with 6 clubs having won the title. Overall, a total of 23 clubs have won the title (only 6 from 2010). The competition has thus been dominated by clubs from just four out of UEFA's 54 member nations, and this dominance is typically associated with the size of football markets in these countries, that allow the big clubs to buy up the best players.% (Figure 1) relates the average league position of teams appearing in the top two English divisions over a decade to their average wage spending, expressed in proportion to the average wage spending of all clubs. The figure shows that wage spending is highly correlated with league position over time. The ÃÂ¢ÃÂÃÂR2ÃÂ¢ÃÂÃÂ is a statistic which measures the percentage of variation in league position that is captured by wage spending, so a value of 90% is a very strong signal of a significant relationship between wages and league standing. Likewise, the relationship between average league position and revenues, again expressed relative to the average of all other clubs, is very close indeed. Over a decade the R2 is again just over 90%. These correlations need explanation ÃÂ¢ÃÂÃÂÃÂÃÂ­ a theory of how the market works ÃÂÃÂ­ÃÂ¢ÃÂÃÂ in order to be understood.
 
Participation in this competition can represent a significant financial boost for clubs. In particular, performances in the tournament, such as advancement from round to round, largely increase club's payouts. In the 2021/2022 season, for instance, the winning club (Real Madrid) collected \euro 83.2 million in performance-based prize money (including an additional \euro 4.5 million for winning the UEFA Super Cup in August 2022). %https://www.sportingnews.com/us/soccer/news/champions-league-prize-money-2022-2023-ucl-winners-uefa/axbbtipavsvy1howxwj6vanp 
While 55 percent of the total prize is distributed according to such performance pay, the rest is distributed among the 32 participants via coefficient payout (30 percent) and broadcast market payout (15 percent).\footnote{The coefficient payout is the outcome of an algorithm that ranks European performance over a 10-year period.} 

As only a few clubs from each domestic competition participate in the UEFA Champions League each year, the scale of such rewards could possibly distort domestic competitions. By treating the substantial increase in payouts to the participating clubs in 1999/2000 as a natural experiment, Pawlowski et al. (2010) provide some \textit{empirical} evidence for this claim. In this paper, we \textit{theoretically} explore the way in which domestic competitions interact with international competitions like the UEFA Champions League. Our focus is on the role that prizes at the international level have on the level of \textit{competitive balance} in the corresponding domestic competitions, with a special emphasis on the implications that our results provide on the regulatory policies for prize schemes in international competitions.

In our (theoretical) analysis, we order clubs by their budget size, assuming that better teams have larger budgets. We also assume that participation and prizes in international competitions are also determined by budget size. 
We acknowledge that this is not always realistic as, sometimes, small-budget clubs participate in international competitions (and receive prizes) and some larger-budget clubs may fail to compete internationally. Nevertheless, in reality, larger-budget clubs indeed tend to perform better and compete internationally most of the times.\footnote{Kuper and Szymanski (2018) show that wage spending is highly correlated with league position over time. Nevertheless, it is also argued that market size highly correlates with regulatory changes in the 1990s (e.g., Simmons, 1997; Frick, 2009).}
%in the "weak league" extension, it seems to matter (regarding total revenue) if a large or small budget team gets the prize money. In reality, , but it is not true that This assumption is sensitive to our results. 

\section{The benchmark analysis}
%A stylized model

Consider a set of domestic (national) league competitions, one for each country, i.e., Bundesliga in Germany, La Liga in Spain, Premier League in England, Serie A in Italy etc. Each competition involves $n \geq 2$ clubs. Let $N = \{1,\ldots,n\}$. %playing a double round-robin tournament.\footnote{Normally, $n=20$, but other values might also exist, e.g., $n=18$ in Bundesliga. As we shall mention later, there are some exceptions within Europe to this format.} 
For each $i\in N$, let $x_i$ denote the budget of club $i$. Let $X=(x_{i})_{i\in N}=(x_1,\ldots x_n)$ and $x=\sum_{i=1}^n x_i$. %where a win in each game is awarded 3 points and each tie is awarded 1 point. 
For ease of exposition, assume $x_1\ge x_2\ge \ldots\ge x_n$. We assume that clubs are utility maximizers (e.g., Sloane, 1971), i.e., they invest their budgets in order to maximize sporting success. As such, budgets determine club's strength.\footnote{We acknowledge that, empirically, club investments are correlated with expected payoffs. But we dismiss this aspect from our stylized theoretical analysis.
%Modeling investments as a function of expected payoffs would have empirical implications and perhaps more importantly, it could yield more interesting results. For example, it seems possible that the very top teams are nearly certain they will get a payment from the international league, but the mediocre teams need to improve to get a payment. In this case, the incentive stemming from the international league could cause the mediocre teams to invest more and thereby change competitive balance. In other words, the payment from the international league could influence the initial budgets, which would change competitive balance.
}

We explore the level of competitive balance of competition $N$ by focussing on the budget distribution, via the so-called Herfindahl-Hirschman Index. The index \textit{increases} with \textit{decreasing} levels of competitive balance. Formally,\footnote{%As we are only interested in relative values of $H$ in the same competition, 
To ease our analysis, we ignore the $10000$ term that is typically considered multiplying the expression.}
$$
H(X)=\sum_{i\in N}\left(\frac{x_i}{x}\right)^2=\frac{1}{x^2}\sum_{i\in N}x_i^2.
$$

We also assume that there exists an international competition with a prize structure granting an overall endowment $E$ to the domestic competition. Let $R$ be a sharing rule allocating $E$ among clubs in $N$ which yields the vector of shares $(R_1(E),\ldots, R_n(E))$ to each  club with $\sum_{j=1}^n R_j(E) =E$.%\footnote{Our model builds on the literature about fair allocation rules (e.g., Thomson, 2011), which has also recently paid attention to issues related to professional sports (e.g., Berganti\~ nos and Moreno-Ternero, 2020, 2021; Krumer and Moreno-Ternero, 2023). It also connects to the literature on prize allocation in contests (e.g., Moldovanu and Sela, 2001; Dietzenbacher and Kondratev, 2023). %which is also largely connected to sports economics (e.g., Sela, 2023).}

 We want to explore the implications of sharing rules into competitive balance of domestic competitions. To do so, let
$x_j^{R,E}=x_j+R_j (E)$, for each $j=1,\dots n$. The aim is to compare the value of the competitive balance index introduced above for distributions $X$ and $X^{R,E}$.%we need to consider the resulting budget distribution $X^{R,E}$ after implementing $R$ for endowment $E$, i.e., 
%should we add a paragraph (before we start with "our" sharing rules) that exactly states how it is done in practice? If we would do so, we could relate to it in the following (when ever it fits). As such, we could, for instance, write whether/how the small "solidarity payment" which is in practice changes (or not) our results. We would need overviews on payouts, right? Like this one on pp. 41/42: https://editorial.uefa.com/resources/0275-151e1a55c231-ef1c32b881dc-1000/en_ln_uefa_financial_report_2020-2021.pdf

%\section{Sharing Rules $R^k$}

The first family of sharing rules we consider  are the  rules $R^k$ in which  the $k$-top clubs in each domestic competition are (equally) rewarded, while the others receive zero. 
%An interesting policy implication of our analysis is that the prize structure of the international competition should be modified. As performance-pay is a deeply rooted principle in sports, a likely suggestion is to modify the structure from club-based to domestic competition-based. 
That is, rule $R^k$ is such that $R_j^k(E)=\frac{E}{k}$ for each $j\in \{1,\ldots, k\}$, and $0$ otherwise. To simplify the notation we denote the resulting distribution by $X^{k,E}$.

To begin, we consider the  rule $R^1$ that  allocates the endowment to a unique club in the domestic competition. That is, $R^1$ is a top performance-based prize structure of the international competition.
%It is just a quick idea and I do not know whether it fits here or elsewhere, but an interesting idea could be to have reward based payments directly to the clubs and to have association-based payments for the market pool in the future; This would also (i) distort CB - however, to a much lesser extent; It would keep (ii) winning incentives for clubs in UCL and (iii) it would reward all domestic clubs for their contribution that soccer is so attractive in the domestic market (which contributes to a large market pool in a given country) 
Formally, $R^1$ is such that $R_1^1(E)=E$, $R_j^1(E)=0$ for each $j> 1$. 
We then obtain the following result. 

\begin{proposition}\label{prizes}
The top performance-based prize structure of the international competition $R^1$ hurts competitive balance at domestic competitions.  That is,  $H(X^{1,E}) \ge H(X)$.
\end{proposition}
\begin{proof}
%For each $i\in N$, let $x_i$ and $\widehat{x}_i$ denote the budget of team $i$, before and after $R^1$ is implemented. 
As mentioned above, in the proofs that follow we shall ignore the constant in the definition of the HHI for being innocuous in our analysis. Then,
$$
H(X^{1,E})=\frac{(x_1+E)^2+\sum_{i\in N\setminus\{1\}}x_i^2}{(x+E)^2}.
$$
Thus, $H(X^{1,E})\ge H(X)$ if and only if
$$
\frac{(x_1+E)^2+\sum_{i\in N\setminus\{1\}}x_i^2}{(x+E)^2}\ge\frac{\sum_{i\in N}x_i^2}{x^2}.
$$
That is, 
$$
%x^2\left(2Ex_1+E^2+\sum_{i\in N}x_i^2\right)=
x^2\left((x_1+E)^2+\sum_{i\in N\setminus\{1\}}x_i^2\right)\ge(x+E)^2 \sum_{i\in N}x_i^2.
$$
Or, equivalently, 
$$
(2x_1+E)x^2\ge(2x+E)\sum_{i\in N}x_i^2,
$$
which is true as $x_1\ge x_i$, for each $i\in N$.%\footnote{The formal proof is tedious but doable, based on my handwritten notes.}
\end{proof}
\medskip

We now consider an alternative rule in which we move away from top club-based to domestic competition-based performance-pay. 
%An interesting policy implication of our analysis is that the prize structure of the international competition should be modified. As performance-pay is a deeply rooted principle in sports, a likely suggestion is to modify the structure from club-based to domestic competition-based. 
That is, instead of just rewarding the top club, an equal-sharing of the international competition prize is imposed among all clubs in the domestic competition. Formally, assume that rule $R^n$ is such that $R_i^n(E)=\frac{E}{n}$ for each $i\in N$. 

\begin{proposition}\label{prizes2}
The equal-sharing structure of the international competition prize among clubs in the domestic competition $R^n$ improves competitive balance at domestic competitions.  That is,  $H(X^{n,E}) \le H(X)$.
\end{proposition}
\begin{proof}
We have
$$
H(X^{n,E})=\frac{\sum_{i\in N}(x_i+\frac{E}{n})^2}{(x+E)^2}.
$$
Then, $H(X^{n,E})\le H(X)$ if and only if
$$
%\frac{\sum_{i\in N}(x_i+\frac{E}{n})^2}{(x+E)^2}=
\frac{\sum_{i\in N}x_i^2}{(x+E)^2}+\frac{E^2}{n(x+E)^2}+\frac{2EX}{n(x+E)^2}\le\frac{\sum_{i\in N}x_i^2}{x^2}.
$$
That is, 
$$
\sum_{i\in N}x_i^2+\frac{E^2}{n}+\frac{2Ex}{n}\le \left(\frac{x+E}{x}\right)^2\sum_{i\in N}x_i^2.
$$
Or, equivalently, 
$$
\frac{E^2}{n}+\frac{2Ex}{n}\le %\left(\left(1+\frac{E}{x}\right)^2-1\right)\sum_{i\in N}x_i^2=\left(\frac{E^2}{x^2}+\frac{2E}{x}\right)\sum_{i\in N}x_i^2=\
\frac{E}{x^2}\left(R+2x\right)\sum_{i\in N}x_i^2,
$$
which translates into
$$
x^2\le n\sum_{i\in N}x_i^2.
$$
As $\varphi(\frac{\sum_i x_i}{n})\le\frac{\sum_i \varphi(x_i)}{n}$, for each convex function $\varphi$ (this is the so-called Jensen inequality), by setting $\varphi(y)=y^2$, we confirm the validity of the last inequality.\end{proof}
\medskip
%A similar argument to the one at the proof of Proposition 1 allows us to prove Proposition 2. 

After examining the two extreme cases, where $k=1$ and $k=n$, we turn to the intermediate rules $R^k$ in which  the $k$-top clubs in each domestic competition are (equally) rewarded, with $1 \le k \le n$, while the others receive zero.  We first establish that while the effect of awards on competitive balance is not uniform,  there are necessary and sufficient conditions that identify when competitive balance is reduced or increased. 
%An interesting policy implication of our analysis is that the prize structure of the international competition should be modified. As performance-pay is a deeply rooted principle in sports, a likely suggestion is to modify the structure from club-based to domestic 

\begin{lemma}\label{prizesk}
The $k$-top performance structure of the international competition prize among clubs in the domestic competition $R^k$ improves competitive balance at domestic competitions for some configurations. More precisely,  $H(X^{k,E})\le H(X)$  if and only if the following condition holds:
$$
x^2\left(E+2\sum_{i=1}^k x_i\right)\le k\left(E+2x\right)\sum_{i\in N}x_i^2.
$$
\end{lemma}
\begin{proof}
We have 
$$
H(X^{k,E})=\frac{\left(\sum_{i=1}^k(x_i+\frac{E}{k})^2\right)+\sum_{i=k+1}^n x_i^2}{(x+E)^2}=\frac{\sum_{i=1}^n x_i^2}{(x+E)^2}+\frac{E^2}{k(x+E)^2}+\frac{2E\sum_{i=1}^k x_i}{k(x+E)^2}.
$$
Thus, $H(X^{k,E})\le H(X)$ if and only if
$$
\frac{\sum_{i\in N} x_i^2}{(x+E)^2}+\frac{E^2}{k(x+E)^2}+\frac{2E\sum_{i=1}^k x_i}{k(x+E)^2}\le\frac{\sum_{i\in N}x_i^2}{x^2}.
$$
That is, 
$$
\sum_{i\in N}x_i^2+\frac{E^2}{k}+\frac{2E\sum_{i=1}^k x_i}{k}\le \left(\frac{x+E}{x}\right)^2\sum_{i\in N}x_i^2,
$$
Or, equivalently, 
$$
\frac{E}{k}\left(E+2\sum_{i=1}^k x_i\right)%=\frac{E^2}{k}+\frac{2E\sum_{i=1}^k x_i}{k}
\le %\left(\left(1+\frac{E}{x}\right)^2-1\right)\sum_{i\in N}x_i^2=\left(\frac{E^2}{x^2}+\frac{2E}{x}\right)\sum_{i\in N}x_i^2=
\frac{E}{x^2}\left(E+2x\right)\sum_{i\in N}x_i^2,
$$
which translates into
$$
x^2\left(E+2\sum_{i=1}^k x_i\right)\le k\left(E+2x\right)\sum_{i\in N}x_i^2,
$$
as stated. 
\end{proof}
%What can we say for the intermediate values of $k$?

As shown in Propositions 1 and 2, the inequality in the statement of Lemma 1 holds for $k=n$, whereas it reverses for $k=1$. The next lemma helps establishing what happens in the intermediate cases. 
%Now let $R>0$ be a prize. Consider first a mechanism that awards $R/k$ to $k$ leading teams. That is,
Recall that:
\[H(X^{k,E})=\frac{\sum_{i=1}^k (x_i+\frac{E}{k}) ^2 + \sum_{i=k+1}^n x_i^2}{(x+E)^2}.\]
\begin{lemma}\label{monot}
    For each $X$ and each $E$, the function  $H(X^{k,E})$  is decreasing in $k$. %That is, the competitiveness is increasing  if the prize is shared among larger number of teams.  
\end{lemma}
\begin{proof}
Let us show that for each $k$, $H(X^{k-1,E})> H(X^{k,E})$. 
By ignoring  common terms and factors on both sides, it suffices to show that
\[
\sum_{i=1}^{k-1} (x_i + \frac{E}{k-1})^2 +x_k^2 > \sum_{i=1}^{k} (x_i + \frac{E}{k})^2,
\]
or
\[
\frac{ 2E \sum_{i=1}^{k-1} x_i }{k-1} + \frac{E^2}{k-1}> \frac{2E \sum_{i=1}^{k} x_i}{k}+ \frac{E^2}{k}.
\]
The latter inequality  holds in our framework as the average budget of the  top  $k-1$ clubs exceeds that 
of the top $k$ clubs.%$\Box$
\end{proof}
\medskip

To summarize, Lemma \ref{monot} implies that $H(X^{k,E})$ is decreasing in $k$. Propositions 1 and 2 show that $H(X^{n,E})\le H(X)\le H(X^{1,E})$. Thus, it follows that there exists $k^{\ast} (X,E) \in\{1,\ldots,n\}$ such that $H(X^{n,E})\le H(X^{k^{\ast}+1,E}) \le H(X)\le H(X^{k^{\ast},E})\le H(X^{1,E})$. In other words, competitive balance of the status-quo budget distribution (with no prize allocation) single crosses the sequence of resulting budget distributions after implementing the $R^k$ rules.\footnote{The so-called \textit{single-crossing property} has some important implications, well established in the public economics literature, referring to the progressivity comparisons of schedules as well as the identification of the majority voting equilibrium (e.g., Jakobsson, 1976; Hemming and Keen, 1983; Gans and Smart 1996; Moreno-Ternero, 2011; Berganti\~{n}os and Moreno-Ternero, 2023b).} 
Equivalently, for each status-quo budget distribution $X$ there exists a threshold $k^{\ast}(X,E)$, such that $H(X^{1,E})>H(X^{2,E})>...>H(X^{k^{\ast}-1},E)>H(X)>H(X^{k^{\ast},E})>...>H(X^{n,E})$. In other words, we have the following: 
\begin{proposition}\label{prizes3}
For each budget distribution $X$ there exists a threshold $k^{\ast}(X,E)\in\{1,\ldots n\}$, such that %$H(X^1)>H(X^2)>...>H(X^{k^{\ast}-1})>H(X)>H(X^{k^{\ast}})>...>H(X^n)$ The 
\begin{itemize}
\item For each $k=1,\ldots, k^{\ast}(X,E)-1$, the $k$-top performance structure of the international competition prize among clubs in the domestic competition $R^k$ hurts competitive balance at domestic competitions (with respect to the original distribution $X$).
\item For each $k=k^{\ast}(X,E),\ldots, n$, the $k$-top performance structure of the international competition prize among clubs in the domestic competition $R^k$ improves competitive balance at domestic competitions (with respect to the original distribution $X$).
\end{itemize}
\end{proposition}

As stated in Proposition 3, there will be a threshold for $k$ (depending on the distribution $X$ and prize $E$) separating the $k$-top rules that hurt competitive balance at domestic competitions from those that improve it. The $k$-top rules are fully ranked according to that feature, ranging from the $1$-top rule (that hurts competitive balance for all distribution $X$) to the $n$-top (equal-sharing) rule (that always improves competitive balance for all distribution $X$). 

%We will turn to  further examination of our results.

\section{Extension and further analysis} 

An interesting corollary from the above discussion is that competitive balance is reduced  much more in ``weak leagues", i.e., leagues for which only one club participates in the international competition.\footnote{As mentioned in Section 2, some leagues might end up having no participant whatsoever in the group stage of the UEFA Champions League, but all leagues have at least one slot in the qualifying tournament.} For those leagues, $R^1$ is the only available $k$-top rule. And we know from Proposition 1 that such a rule always hurts competitive balance, whereas other $k$-top rules might not. As stated in Lemma 2, the more clubs of a league benefit from the prize of the international competition, the better competitive balance in the domestic competition is. In other words: 

\begin{corollary}\label{small-leagues}
Competitive balance in strong domestic leagues with more spots in the international competition is not hurt as much as in weak leagues with just one spot in the international competition.
\end{corollary}

We now generalize $k$-top performance schemes to allow for uneven rewards among top clubs. Formally, let $k\in \{1,\ldots, n\}$ and $a\in \Bbb{R}^k$ be such that $a_1\ge a_2\ge\ldots\ge a_k>0$, $a_1>\frac{1}{k}$, and $\sum_{i=1}^{k} a_i=1$. Then, the rule $R^{a,k}$ is such that $R^{a,k}_i(E)=a_i E$ for each $i\in \{1,\ldots, k\}$ and $0$ otherwise. Our next result states that the even $k$-top performance scheme outperforms (in terms of competitive balance for the domestic competition) all the uneven $k$-top performance schemes.

\begin{proposition}\label{uneven-k}
The even $k$-top performance structure of the international competition prize among clubs in the domestic competition $R^k$ improves competitive balance at domestic competitions more than any other uneven $k$-top performance structure. 
\end{proposition}
\begin{proof}
For each $i\in N$, let $x^k_i=x_i+\frac{E}{k}$ and $x^{a,k}_i=x_i+a_iE$ denote the budget of club $i$, after $R^k$ and $R^{a,k}$ are implemented. Then, 
$$
H(X^k)=\frac{\sum_{i=1}^k(x_i+\frac{E}{k})^2+\sum_{i=k+1}^n x_i^2}{(x+E)^2}=\frac{\sum_{i=1}^k x_i^2}{(x+E)^2}+\frac{E^2}{k(x+E)^2}+\frac{2E\sum_{i=1}^k x_i}{k(x+E)^2}+\frac{\sum_{i=k+1}^n x_i^2}{(x+E)^2},
$$
and 
$$
H(X^{a,k})=\frac{\sum_{i=1}^k(x_i+a_iE)^2+\sum_{i=k+1}^n x_i^2}{(x+E)^2}=\frac{\sum_{i=1}^k x_i^2}{(x+E)^2}+\frac{E^2\sum_{i=1}^k a^2_i}{(x+E)^2}+\frac{2E\sum_{i=1}^k a_ix_i}{(x+E)^2}+\frac{\sum_{i=k+1}^n x_i^2}{(x+E)^2},
$$
Thus, $H(X^k)\le H(X^{a,k})$ if and only if
$$
\frac{\sum_{i=1}^k x_i^2}{(x+E)^2}+\frac{E^2}{k(x+E)^2}+\frac{2E\sum_{i=1}^k x_i}{k(x+E)^2}\le\frac{\sum_{i=1}^k x_i^2}{(x+E)^2}+\frac{E^2\sum_{i=1}^k a^2_i}{(x+E)^2}+\frac{2E\sum_{i=1}^k a_ix_i}{(x+E)^2}.
$$
That is, 
$$
\frac{E^2}{k}+\frac{2E\sum_{i=1}^k x_i}{k}\le E^2\sum_{i=1}^k a^2_i+2E\sum_{i=1}^k a_ix_i,
$$
Or, equivalently, 
$$
\frac{E}{k}\left(E+2\sum_{i=1}^k x_i\right)\le E\left(E\sum_{i=1}^k a^2_i+2\sum_{i=1}^k a_ix_i\right),
$$
which translates into
%$$
%E+2\sum_{i=1}^k x_i\le k\left(E\sum_{i=1}^k a^2_i+2\sum_{i=1}^k a_ix_i\right).
%$$
$$
E\left(\sum_{i=1}^k a^2_i-\frac{1}{k}\right)\ge 2\left(\sum_{i=1}^k \left(\frac{1}{k}-a_i\right)x_i\right).
$$
Note that the left hand side is non-negative because of the Jensen inequality.\footnote{If $\varphi$ is convex then $\varphi(\frac{\sum_i a_i}{k})\le\frac{\sum_i \varphi(a_i)}{k}$. Thus, for $\varphi(x)=x^2$, we have $(\frac{1}{k})^2=(\frac{\sum_i a_i}{k})^2\le\frac{\sum_i a^2_i}{k}$, i.e., $\sum_{i=1}^k a^2_i\ge\frac{1}{k}$.} And the right-hand side is actually non-positive because $x_1\ge x_2\ge\ldots\ge x_k$, $a_1\ge a_2\ge\ldots\ge a_k>0$, $a_1>\frac{1}{k}$, and $\sum_{i=1}^{k} a_i=1$. More precisely, let $\varepsilon=a_1-\frac{1}{k}>0$ and, for each $i=2,\ldots k$, let $\varepsilon_i=\frac{1}{k}-a_i>0$. Note that $0<\varepsilon_1\le \varepsilon_2\le\ldots\le \varepsilon_k$, and $\sum_{i=2}^{k} \varepsilon_i=\varepsilon$. Then, $\sum_{i=1}^k \left(\frac{1}{k}-a_i\right)x_i=-\varepsilon x_1+\sum_{i=2}^{k} \varepsilon_i x_i=\sum_{i=2}^{k} \varepsilon_i (x_i-x_1)\le 0$.
\end{proof}
\bigskip

We complete this section upon exploring the impact of the value of the award $E$ on the level of competitive balance, as well as the threshold reflecting the single-crossing property   of $k$-top schemes' competitive balance.  We first claim that for every budget distribution $X$ and the number of clubs $k$ receiving the award, competitive balance as a function of $E$ is single-peaked. That is, there is a threshold $E^{\ast} (X,k)$ that depends on $X$ and $k$ such that the level of competitive balance increases for $E<E^{\ast} (X,k)$ and declines for $E>E^{\ast} (X,k)$. %The conclusion is quite intuitive: for large $E$ the magnitude of award dominates clubs' budgets making the budget distribution more equitable, which is not the case for smaller $E$. Formally,
\begin{proposition}\label{prizes4}
Denote
$$
E^{\ast} (X,k)=\max\left\{0,\frac{k\sum_{i=1}^n x_i^2-x\sum_{i=1}^k x_i}{\sum_{i=k+1}^n x_i}\right\}.
$$
%If $E^{\ast} (X,k)=0$, then $H(X^{k,E})$ is increasing for all $E\ge 0$. If $E^{\ast} (X,k)>0$, 
Then, $H(X^{k,E})$ is decreasing for all $0\le E\le E^{\ast} (X,k)$ and increasing for all $E\ge E^{\ast} (X,k)$. 
% How does $E$ impact $H(X^{k,E})$?\\
\end{proposition}

\begin{proof}
Recall that
$$
H(X^{k,E})=\frac{\left(\sum_{i=1}^k(x_i+\frac{E}{k})^2\right)+\sum_{i=k+1}^n x_i^2}{(x+E)^2}=\frac{\sum_{i=1}^n x_i^2}{(x+E)^2}+\frac{E^2}{k(x+E)^2}+\frac{2E\sum_{i=1}^k x_i}{k(x+E)^2}.
$$
Taking derivatives of this expression with respect to $E$, we obtain
$$
\frac{2}{k(x+E)^3}\left( (x+E)(\sum_{i=1}^k x_i+E)-k\sum_{i=1}^n x_i^2-E^2-2E\sum_{i=1}^k x_i\right).
$$
Or, equivalently,
$$
\frac{2}{k(x+E)^3}\left( E\sum_{i=k+1}^n x_i-k\sum_{i=1}^n x_i^2+x\sum_{i=1}^k x_i\right),
$$
from where it follows that $H(X^{k,E})$ is increasing with respect to $E$ if and only if
$$
E\ge\frac{k\sum_{i=1}^n x_i^2-x\sum_{i=1}^k x_i}{\sum_{i=k+1}^n x_i}.
$$
Note that the numerator of the last expression  $k\sum_{i=1}^n x_i^2-x\sum_{i=1}^k x_i$ could be either negative or positive for various $k$. Indeed, for $k=n$, $n\sum_{i=1}^n x_i^2\ge(\sum_{i=1}^n x_i)^2= x^2$. Whereas for $k=1$, we have $\sum_{i=1}^n x_i^2 - x x_1 = \sum_{i=1}^n [ x_i^2 -x_1 x_i]\le 0.$  Thus, by setting  $E^{\ast} (X,k)=\max\{0,\frac{k\sum_{i=1}^n x_i^2-x\sum_{i=1}^k x_i}{\sum_{i=k+1}^n x_i}\}$, we conclude the proof.
\end{proof}
\bigskip

Finally, how does $E$ impact $k^{\ast}(X,E)$? According to Lemma \ref{prizesk}, $H(X^{k,E})\le H(X)$  if and only
$$
x^2\left(E+2\sum_{i=1}^k x_i\right)\le k\left(E+2x\right)\sum_{i\in N}x_i^2.
$$
Equivalently, 
$$
\left(k\sum_{i=1}^n x_i^2-x^2\right)E\ge 2x\left(x\sum_{i=1}^k x_i-k\sum_{i=1}^n x_i^2\right).
$$

If $k\sum_{i=1}^n x_i^2\ge x^2$, then the condition above trivially holds for all $E$ (the left-hand side would be non-negative whereas the right-hand side would be non-positive). 

If $k\sum_{i=1}^n x_i^2\le x\sum_{i=1}^k x_i$, then the condition above trivially does not hold for any $E$ (the left-hand side would be non-positive whereas the right-hand side would be non-negative). 

If $x\sum_{i=1}^k x_i< k\sum_{i=1}^n x_i^2< x^2$, then the condition is equivalent to write the following:
$$
E\le \frac{2x\left(k\sum_{i=1}^n x_i^2-x\sum_{i=1}^k x_i\right)}{x^2-k\sum_{i=1}^n x_i^2}.
$$
%So, the larger $E$, the less likely the condition $H(X^{k,E})\le H(X)$ is satisfied and, thus, the larger $k^{\ast}(X,E)$ is. 

We now explore the above upper bound for $E$, starting with its monotonicity with respect to $k$.

\begin{lemma}\label{monot2}
    Assume $x\sum_{i=1}^k x_i< k\sum_{i=1}^n x_i^2< x^2$ and denote 
    $$
    \hat{E}(k,X)=\frac{2x\left(k\sum_{i=1}^n x_i^2-x\sum_{i=1}^k x_i\right)}{x^2-k\sum_{i=1}^n x_i^2}.
    $$
   Then, for each $X$ and each $E$, the function  $\hat{E}(k,X)$  is increasing in $k$. %That is, the competitiveness is increasing  if the prize is shared among larger number of teams.  
\end{lemma}
\begin{proof}
Assume $x\sum_{i=1}^k x_i< k\sum_{i=1}^n x_i^2< x^2$. Let us show that for each $k$, $\hat{E}(k+1,X)\ge\hat{E}(k,X)$. 
That is,
$$
\frac{2x\left((k+1)\sum_{i=1}^n x_i^2-x\sum_{i=1}^{k+1} x_i\right)}{x^2-(k+1)\sum_{i=1}^n x_i^2}\ge \frac{2x\left(k\sum_{i=1}^n x_i^2-x\sum_{i=1}^k x_i\right)}{x^2-k\sum_{i=1}^n x_i^2}.
$$
Equivalently, 
$$
\left(x^2-k\sum_{i=1}^n x_i^2\right)\left((k+1)\sum_{i=1}^n x_i^2-x\sum_{i=1}^{k+1} x_i\right)\ge
\left(x^2-(k+1)\sum_{i=1}^n x_i^2\right) \left(k\sum_{i=1}^n x_i^2-x\sum_{i=1}^k x_i\right).
$$
By ignoring  common terms and factors on both sides, it suffices to show that
\[
\left(x^2-k\sum_{i=1}^n x_i^2\right)\left(\sum_{i=1}^n x_i^2-x x_{k+1}\right) \ge -\left(\sum_{i=1}^n x_i^2\right)\left(k\sum_{i=1}^n x_i^2-x\sum_{i=1}^k x_i\right),
\]
As $x^2\ge k\sum_{i=1}^n x_i^2$ and $k\sum_{i=1}^n x_i^2\ge x\sum_{i=1}^k x_i$ due to the premises of the lemma, it suffices to show that 
\[
\sum_{i=1}^n x_i^2\ge x x_{k+1}=x_{k+1}\sum_{i=1}^n x_i.
\]
Now, the premises of the lemma also imply that $\sum_{i=1}^n x_i^2\ge x \sum_{i=1}^k \frac{x_i}{k}$. As the average budget of the top $k$ clubs exceeds that 
of the $k+1$ club (i.e., $\sum_{i=1}^k \frac{x_i}{k}\ge x_{k+1}$), the above inequality holds, as desired.%$\Box$
\end{proof}
\medskip

Let us then assume that the premises of Lemma 3 hold, i.e., $x\sum_{i=1}^k x_i< k\sum_{i=1}^n x_i^2< x^2$. Then, $k^{\ast}(X,E)$ is the unique value satisfying 
$$
\hat{E}(k^{\ast},X)\le E \le\hat{E}(k^{\ast}+1,X). 
$$
Thus, if $E$ increases then either the above upper bound keeps holding (which means that $k^{\ast}$ remains the same) or not (which means that $k^{\ast}$ increases to at least $k^{\ast}+1$ or beyond). Thus, if $E$ increases then $k^{\ast}$ weakly increases. This also trivially extends to the case in which the premises of Lemma 3 do not hold (as then either $k^{\ast}=1$ or $k^{\ast}=n$). In other words, we have the following:

\begin{proposition}\label{prizes4}
The optimal value $k^{\ast}$ weakly increases with $E$.
\end{proposition}

\section{An illustration}
We now illustrate our results for a special case. We take as starting point the squad spending limits in La Liga (the highest professional soccer league in Spain) for the 2023/2024 season.\footnote{
https://www.statista.com/statistics/764962/budget-equipment-from-football-from-the-league-in-spain-2015-2016/} The squad spending limit denotes the maximum amount that each club or public limited sports company (SAD) can spend during the course of the summer season, including costs related to players, the manager, assistant manager and fitness coach. It also includes spending on the reserves, the youth system and other areas. As clubs usually spend all they are allowed, we can consider these spending limits as equivalent to budgets. 

We start noticing that the HHI for the initial distribution is $1232.4$. Thus, according to the U.S. Horizontal Merger Guidelines, within the \textit{moderately concentrated} region. It is important to note that Barcelona's limit spending (see Table 1) reflects a significant decrease from the ceiling established a few months earlier due to the diminishing impact of the assets they sold off to improve short-term finances. Without that decrease, the HHI for the initial distribution would have been much higher, actually very close to the border between the \textit{moderately concentrated} region and the \textit{highly concentrated} region.% when La Liga notified all clubs of their compensation levels for the upcoming five months when they would be reviewed again after the conclusion of the January

We now apply our results and obtain that, for each $k=1,2,\dots, 5$, the $k$-top rules hurt competitive balance upon increasing the value of HHI, for each possible value of the prize endowment $E$. The extent of this increase depends on the specific rule and the specific value of the endowment. For instance, if $k=1$ and $E\ge 425$ (million euros), the HHI of the resulting distribution surpasses $1800$; thus, becoming part of the \textit{highly concentrated} region. 

Note that La Liga is guaranteed $4$ slots in the UEFA Champions League each year, which sometimes becomes $5$ (when the UEFA Champions League winner or the UEFA Europa League winner is also from La Liga, and did not finish within the Top 4 in the domestic competition), but never more than that. Thus, we can conclude that allocating UEFA Champions League prizes to La Liga by rewarding only its clubs that qualified for the UEFA Champions League will certainly hurt La Liga's competitive balance. In other words, to prevent a decrease of competitive balance in La Liga, the prizes from the UEFA Champions League should be redistributed to reach clubs that do not qualify for the UEFA Champions League.

We also obtain that, for each $k=9,10,\dots, 20$, the $k$-top rules improve competitive balance upon decreasing the value of HHI, for each possible value of the prize endowment $E$. The extent of this decrease also depends on the specific rule and the specific value of the endowment. For instance, if $k=20$ and $E\ge 540$ (million euros), the HHI of the resulting distribution falls short of $1000$; thus, within the \textit{unconcentrated} region. 

Finally, we obtain that for each $k=6,7,8$, the $k$-top rules improve competitive balance upon decreasing the value of HHI, for each possible value of the prize endowment $E$ within a realistic domain.\footnote{More precisely, %for $E$ lower than $1264$, $5544$, and $86338$ million euros, respectively. This means that, 
when $E<1264$ (millions), the $6$-top rule increases competitive balance (the HHI value is lower). Likewise, when $E<5544$ (millions), the $7$-top rule increases competitive balance (the HHI value is lower). Finally, when $E<86338$(millions), the $8$-top rule increases competitive balance (the HHI value is lower).} Thus, in the parlance of Proposition 3, we obtain that for the case of La Liga, $k^{\ast}=6$.

\begin{equation*}
    \text{Insert Table 1 about here}
    \end{equation*}
    
 \begin{figure}[h]
  \centering
  \includegraphics{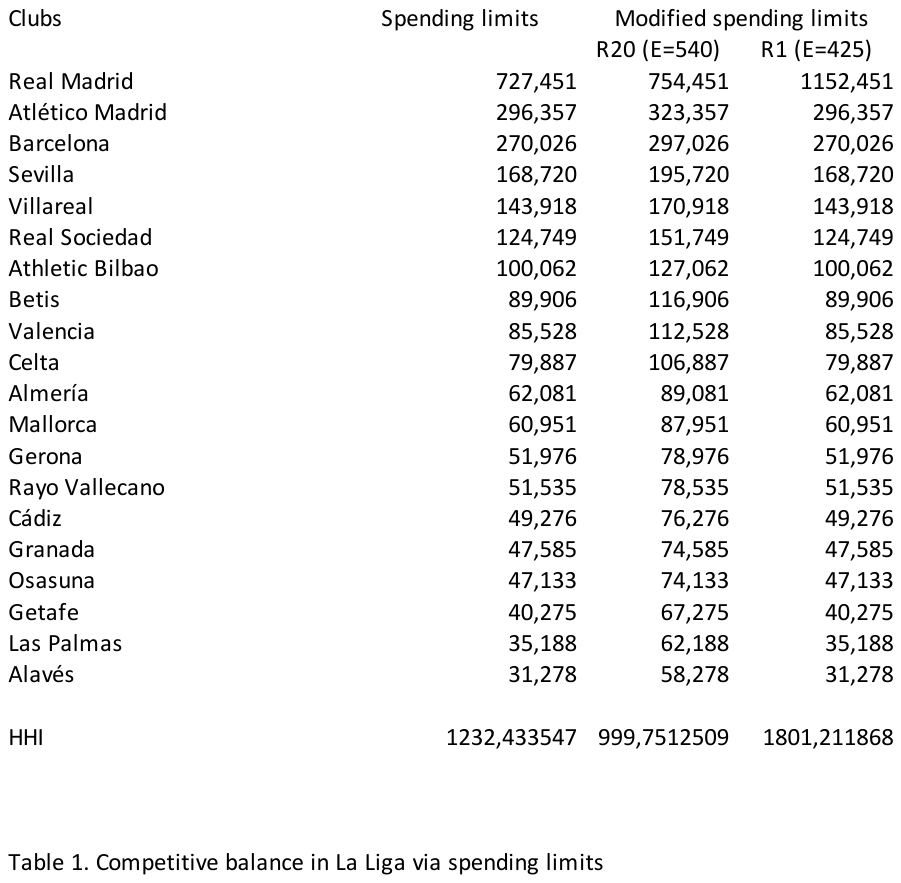}
%  \includegraphics[scale=0.65]{TABLE1.pdf}
  %\caption{Competitive balance in La Liga via spending limits. \label{Table1}}
\end{figure}

\section{Discussion}
We have explored in this paper the effect that the prize allocation of international (sports) competitions has on the structure of domestic competitions. For that purpose, we have resorted to the well-known Herfindahl-Hirschman index. By comparing the value of the index before and after the clubs get prize money for participation in international tournaments, one can analyze the changes in competitive balance. Our main insights can be summarized in the next three items.
\begin{itemize}
\item For each initial clubs' budget distribution in the league, rewarding a high number of clubs with the same prize improves competitive balance, while distributing prize money to a low number of clubs worsens competitive balance.
\item For each initial clubs' budget distribution in the league, equal distribution of prize money between the clubs participating in the international competitions is better (from the competitive balance perspective) than any other monotonic (with respect to a position in the league) distribution between the same clubs.
\item There exists a threshold of the prize value %partitioning the regions with 
such that competitive balance increases (as a function of the prize value) below the threshold, whereas it declines (as a function of the prize value) above the threshold.
\end{itemize}

In other words, as expected, the rule awarding only the top club in the domestic competition unambiguously hurts competitive balance in the domestic competition, whereas the rule equally awarding all clubs in the domestic competition unambiguously improves competitive balance in the domestic competition. Our more insightful results show that compromises between these two polar schemes offer intermediate results. More precisely, the rules that award a subset of the top clubs in each domestic competition have interesting features. On the one hand, they are all fully ranked with respect to competitive balance. That is, the larger the number of top clubs considered, the better. This implies that ``weaker" leagues with only one spot in the international competition are hurt more than the other competitions. On the other hand, although fully ranked according to competitive balance, the fact that the resulting rule hurts or improves competitive balance with respect to the status quo (without prize allocation) depends on the initial conditions (the pre-allocation budget distribution and the size of the prize to be allocated). By extending our benchmark analysis we find support for an \textit{equal-share} scheme compared to \textit{uneven} prize allocations. 

One might infer from the above that UEFA's idea to share at least the market pool more equally amongst participants from the same country seems to be a good idea. %\footnote{$50\%$ depends on the number of participating teams; We knew this, but it is not exactly equal sharing because the better the participating team was in the season before, the more the team gets, i.e., (e.g.) for 2 teams $55\%$ (league winner) and $45\%$ (league runner-up).$50\%$ depends on number of matches played (proportionally), i.e., a team winning the CL gets more than another team from the same league.} 
Thus, in a reward sharing scheme compromising between the market pool and pure performance pay, our analysis suggests to prioritize the former. Likewise, the idea to redistribute at least a small share of revenues from Eredivisie club's participation in UEFA club competitions amongst the non-participating clubs in the Dutch first division seems reasonable according to our findings.\footnote{https://eredivisie.eu/news/eredivisie-first-league-to-redistribute-revenues-from-uefa-club-competitions/}

%When Top 5 leagues are less affected than weaker leagues it would be interesting/relevant to explore empirically the effects of CB payouts in smaller leagues (because we already find in 2010 for Top 5 leagues that payout harm CB).} We also show that the optimal value regarding the number of top teams to be awarded only weakly increases in the size of the prize. Moreover, for a given number of top teams awarded, we find that competitive balance follows a U-shape regarding the size of the prize. In other words, competitive balance is declining below some threshold value of the prize and raising above that threshold.

Moreover, we might argue that increasing the number of participating clubs, allocating more than just one spot at least to some leagues, seems reasonable. But it remains an empirical question whether UEFA Champions League revenues have increased (or will increase after format change in 2024) enough to mitigate the effect on competitive balance. Likewise, as we suggest that the Top 5 European leagues (French Ligue 1, Italian Serie A, German Bundesliga, Spanish La Liga and English Premier League) are less affected (because they get more spots in the international competition), it would be interesting to explore empirically the effects in smaller leagues to complement what Pawlowski et al., (2010) found for the Top 5 leagues.

We also believe our work provides relevant insights for the strand of literature on financial sustainability
%https://www.uefa.com/insideuefa/news/0246-0e796c23daa9-41f78afb0c7a-1000--financial-sustainability/ Financial Fair Play 
(e.g., Peeters and Szymanski, 2014). Both regulation policies (distribution of prize money and financial sustainability rules) play a crucial role in determining the future of European football. In both cases, one of the most important questions is what will happen with the ensuing inequality.

We conclude with two caveats. 

First, our analysis has relied on HHI as a proxy for competitive balance, which could be rationalized upon connecting with the classical literature mentioned at the introduction, which relates HHI to welfare via fan interest. But this is contested as some leagues with different levels of competitive balance in recent years (e.g., Bundesliga and Premier League) have stadiums operating at almost full capacity over a season.%The Premier League "stadiums were at 98.7% of capacity in 2022-23"1. Bundesliga sells 92%2. What is the welfare link with the HHI as measured?
\footnote{The uncertainty of outcome hypothesis is nevertheless usually supported with broadcasting figures, rather than stadium attendance.} %some TV viewing studies support UOH (and any analysis on surprise/suspense  which is related  was just conducted using TV viewing figures ecause the  lit, in particular, does not find much support for the UOH? Moreover, big revenue streams come from (international) broadcasting rights rather than attendance figures; As such, managers care more about (TV) viewing anyway
One might argue that regulatory agencies care about ranges (rather than specific values) of HHI. To this end, we have also illustrated that our results can go beyond a linear concern for competitive balance, offering further insights about the extent of HHI changes that prize allocations yield. More precisely, we have illustrated in a specific case (when the initial budget distribution is provided) that alternative prize allocations can also generate shifts in the range regions of HHI (from highly concentrated to moderately concentrated, or even unconcentrated, and vice versa). 

Second, our analysis is based on monotonicity assumptions: better clubs domestically perform better in international competitions, thus obtain higher revenues and get even stronger. In reality, this is not always the case and unexpected results might occur. However, the probability of unexpected results is low. On the other hand, relative differences in budgets usually translate into much smaller relative differences in outcomes (e.g., points scored in a season). 
It is left for future research to study the robustness of our results to surprises. A plausible way to do so would be considering an expectation of the HHI. Alternatively, one could explore a multi-period model, where unexpected results improve competitive balance, but after many seasons the strong clubs still get even stronger. %Though this is not an easy point to tackle, I believe that the robustness of results to unexpected outcomes is an important issue here that comes from real world tournaments.
%In other words: for years people argue that the elite clubs always win and dominate but then Bremen wins in Munich and Leverkusen plays top (even though they might be just Top 5 or so); By looking at budget differences (which are huge even at the very top end between lets say Top1 and Top 5 in Germany) one would expect this cannot be but the correlation between budget and points scored is not linear.

%\newpage

\end{document}